\pdfoutput=1
\RequirePackage{ifpdf}
\ifpdf 
\documentclass[pdftex]{sigma}
\else
\documentclass{sigma}
\fi

\numberwithin{equation}{section}

\newtheorem{Theorem}{Theorem}[section]
\newtheorem{Lemma}[Theorem]{Lemma}
\newtheorem{Proposition}[Theorem]{Proposition}
 { \theoremstyle{definition}
\newtheorem{Remark}[Theorem]{Remark} }

\usepackage{pdflscape}
\usepackage{mathrsfs,bm}

\newcommand{\bma}{{\bm{a}}}
\newcommand{\bmb}{{\bm{b}}}
\newcommand{\bmc}{{\bm{c}}}
\newcommand{\bme}{{\bm{e}}}

\newcommand{\bmu}{{\bm{u}}}
\newcommand{\bmv}{{\bm{v}}}
\newcommand{\bmw}{{\bm{w}}}

\newcommand{\bbu}{{\bm{\bar{u}}}}

\newcommand{\mcF}{\mathcal{F}}

\newcommand{\mcH}{\mathcal{H}}
\newcommand{\mcK}{\mathcal{K}}

\newcommand{\mcI}{\mathcal{I}}
\newcommand{\mcP}{\mathcal{P}}

\newcommand{\mscrF}{\mathscr{F}}

\newcommand{\textf}[1]{\text{\small #1}}

\newcommand{\bra}[1]{\bigl (#1\bigr )}

\newcommand{\bc}[1]{[\![#1]\!]}

\newcommand{\pobr}[1]{\left \{#1\right \}}
\newcommand{\dual}[1]{\langle #1 \rangle}

\newcommand{\pmatrx}[1]{\begin{pmatrix} #1 \end{pmatrix}}

\DeclareMathOperator{\ad}{ad}
\DeclareMathOperator{\Vect}{Vect}

\newcommand{\dx}{{\rm d}\bm{x}}

\begin{document}

\allowdisplaybreaks

\newcommand{\arXivNumber}{1906.08388}

\renewcommand{\PaperNumber}{094}

\FirstPageHeading

\ShortArticleName{Bi-Hamiltonian Systems in (2+1) and Higher Dimensions Defined by Novikov Algebras}

\ArticleName{Bi-Hamiltonian Systems in (2+1)\\ and Higher Dimensions Defined by Novikov Algebras}

\Author{B{\l}a\.zej M. SZABLIKOWSKI}

\AuthorNameForHeading{B.M.~Szablikowski}

\Address{Faculty of Physics, Division of Mathematical Physics, Adam Mickiewicz University,\\ ul.~Uniwersytetu Pozna\'nskiego~2, 61-614~Pozna\'n, Poland}
\Email{\href{mailto:bszablik@amu.edu.pl}{bszablik@amu.edu.pl}}

\ArticleDates{Received June 21, 2019, in final form November 21, 2019; Published online November 29, 2019}

\Abstract{The results from the article [Strachan I.A.B., Szablikowski B.M., \textit{Stud. Appl. Math.} {\bf 133} (2014), 84--117] are extended over consideration of central extensions allowing the introducing of additional independent variables. Algebraic conditions associated to the first-order central extension with respect to additional independent variables are derived. As result $(2+1)$- and, in principle, higher-dimensional multicomponent bi-Hamiltonian systems are constructed. Necessary classification of the central extensions for low-dimensional Novikov algebras is performed and the theory is illustrated by significant $(2+1)$- and $(3+1)$-dimensional examples.}


\Classification{37K10; 17B80; 37K30}

\section{Introduction}

In the article \cite{SS} we have presented construction of $(1+1)$-dimensional integrable bi-Hamiltonian systems associated with Novikov algebras. These systems are multicomponent generalizations of the Camassa--Holm equation~\cite{CH} and can be interpreted as Euler equations on the respective centrally extended Lie algebras. On the other hand, the central extension procedure is one of the most effective methods allowing for the construction of $(2+1)$ analogs of $(1+1)$-dimensional systems. However, such procedure is not always possible. All the more, there is very limited number of approaches allowing for the systematic construction of $(3+1)$- and higher-dimensional integrable systems, see the recent survey~\cite{BSe1} and references therein.

Novikov algebras naturally appear in the context of the homogeneous first-order Hamiltonian operators \cite{BN, GD}:
\begin{gather}\label{DN}
		\mcP^{ij} = \big(b^{ij}_k+b^{ji}_k\big)u^k\partial_x + b^{ij}_k u^k_x, 	
\end{gather}
which are a special case of Dubrovin--Novikov operators of hydrodynamic type~\cite{DN,DN2}. The algebraic conditions for~\eqref{DN} to be Hamiltonian means that~$b^{ij}_k$ are structure constants of a~Novikov algebra. The Hamiltonian operators \eqref{DN} define Lie--Poisson structures associated with the so-called translationally invariant Lie algebras that are in one to one correspondence with Novikov algebras. For more information about this and directly related topics see~\cite{SS} and the recent works~\cite{St1,St2,SZ}.

Novikov algebras also play significant role in the theory of multi-dimensional Hamiltonian operators of Dubrovin--Novikov type:
\begin{gather}\label{mDN}
	\mcP^{ij} = \sum_{\alpha=1}^D\big[g^{ij\alpha}(\bmu)\partial_{x^\alpha} + b^{ij\alpha}_k(\bmu) u^k_{x^\alpha}\big],\qquad \bmu = \big(u^1,\ldots,u^n\big),
\end{gather}
which in the non-degenerate case, $\det g^{ij\alpha}(\bmu)\neq 0$, can be reduced by a change of dependent variables to a linear form,\footnote{One-dimensional Dubrovin--Novikov operator ($D=1$) can be reduced to a constant form.} where
\begin{gather*}
		g^{ij\alpha}(\bmu) = \big(b^{ij\alpha}_k + b^{ji\alpha}_k\big)u^k + g^{ij\alpha}_0,\qquad b^{ij\alpha}_k(\bmu) = b^{ij\alpha}_k = \text{const},\qquad g^{ij\alpha}_0 = \text{const},
\end{gather*}
and $b^{ij\alpha}_k$ are structure constants of Novikov algebras that additionally satisfy some compatibility conditions, see \cite{Mokhov1,Mokhov2} and also~\cite{St2}.

The key result of this paper is the use of $(N+1)$-dimensional pencils
\begin{gather}\label{pencil}
\mcP^{ij}_\lambda = \mcP^{ij} + \mcP^{ij}_1 + \lambda \mcP^{ij}_0,\qquad \lambda\in\mathbb{C},
\end{gather}
that are central extension of the first-order homogeneous Hamiltonian operators~\eqref{DN}, for the construction of respective $(N+2)$-dimensional bi-Hamiltonian integrable hierarchies. The extension in \eqref{pencil} is defined through admissible cocycles:
\begin{gather*}
	\mcP^{ij}_\mu = g^{ij}_\mu \partial _x + f^{ij}_\mu \partial^2 _x + h^{ij}_\mu \partial _x^3
	+ \sum_{\alpha=1}^N\eta^{ij}_{\mu,\alpha} \partial _{y^\alpha},\qquad \mu=0,1,
\end{gather*}
where $N$ is the dimension of the linear space spanned by independent cocycles introducing additional spatial variables $y^\alpha$. In Theorem~\ref{centext} we derive the algebraic conditions that must be satisfied on a Novikov algebra by the first-order central extension defined with respect to an additional independent variable. Further, we show that the Killing forms must practically satisfy the same algebraic conditions. Subsequently, we classify such central extensions with respect to low-dimensional Novikov algebras, that have been originally classified in the articles \cite{BM1, BM2, BG} and those that can be extracted from the work~\cite{FLS}. Consequently, we show how to construct $(2+1)$- and, in principle, higher-dimensional bi-Hamiltonian hierarchies on the centrally extended Lie algebras associated to Novikov algebras. This construction is formulated with respect to a~Killing form. Finally, we illustrate the above theory by explicit examples of $(2+1)$- and $(3+1)$-dimensional (dispersive) integrable systems and related bi-Hamiltonian structures.

\section{Novikov algebras}

Let $(\mathbb{A},\cdot)$ be a Novikov algebra,\footnote{In this article all algebraic structures are considered over the field of complex numbers $\mathbb{C}$.} this means that it is right-commutative
\begin{gather*}
(\bma\cdot \bmb)\cdot \bmc = (\bm{a}\cdot \bmc)\cdot \bmb,
\end{gather*}
and left-symmetric (quasi-associative)
\begin{gather*}
(\bma\cdot \bmb)\cdot \bmc - \bma\cdot(\bmb\cdot \bmc) = (\bmb\cdot \bma)\cdot \bmc - \bmb\cdot(\bma\cdot \bmc).
\end{gather*}
The structure constants of a Novikov algebra $\mathbb{A}$, with basis vectors $\bme_1,\ldots,\bme_n$, are given by $b^i_{jk}$ such that
\begin{gather*}
\bme_i\cdot \bme_j = \sum_k b^k_{ij}\bme_k,\qquad \mathcal{B} \equiv \bigg(\sum_k b^k_{ij}\bme_k\bigg),
\end{gather*}
where $\mathcal{B}$ is the characteristic matrix.

By $\mathscr{L}_\mathbb{A}$ we understand the algebra of $\mathbb{A}$-valued (smooth) functions
of (independent) spatial variables $x,y_1,\ldots,y_N$ belonging to the domain $\Omega = \mathbb{S}^1\times\dots\times\mathbb{S}^1$,
\begin{gather*}
	\mathscr{L}_\mathbb{A} := \{\bmu\colon \Omega\rightarrow \mathbb{A} \,|\, \text{$\bmu$ is smooth}\}.
\end{gather*}
As the regular dual algebra $\mathscr{L}_\mathbb{A}^*$ to $\mathscr{L}_\mathbb{A}$ we choose the vector space of $\mathbb{A}^*$-valued functions,
\begin{gather*}
	\mathscr{L}_\mathbb{A}^* := \{\bm{\xi}\colon \Omega\rightarrow \mathbb{A}^* \,|\, \text{$\bm{\xi}$ is smooth}\}.
\end{gather*}
The related duality pairing $\mathscr{L}_\mathbb{A}^*\times\mathscr{L}_\mathbb{A}\rightarrow\mathbb{C}$ is given by
\begin{gather*}
	\dual{\bm{\xi},\bmu} := \int_{\Omega}(\bm{\xi},\bmu)\dx,
\end{gather*}
where the integral is over all spatial variables, $\dx\equiv {\rm d}x{\rm d}y_1\cdots {\rm d}y_N$. By $(\,,\,)\colon \mathbb{A}^*\times\mathbb{A}\rightarrow\mathbb{C}$ we mean the natural pairing between $\mathbb{A}$ and its dual $\mathbb{A}^*$.

The Lie algebra structure on $\mathscr{L}_\mathbb{A}$ is defined, with respect to the distinguished variable $x$, by the bracket
\begin{gather}\label{lb}
	\bc{\bmu,\bmv} : = \bmu_x\cdot \bmv - \bmv_x\cdot \bmu,\qquad \bmu_x\equiv \frac{\partial\bmu}{\partial x}.
\end{gather}
In fact, \eqref{lb} is a Lie bracket on $\mathscr{L}_\mathbb{A}$ iff $(\mathbb{A},\cdot)$ is a Novikov algebra.

\section{Central extensions}

Consider a bilinear form $g\colon \mathbb{A}\times \mathbb{A} \rightarrow \mathbb{C}$, then we will say that:
\begin{itemize}\itemsep=0pt
\item the form satisfies the quasi-Frobenius condition if
\begin{gather*}
		g(\bma,\bmb\cdot \bmc) = g(\bma\cdot \bmc, \bmb),
\end{gather*}
\item the form satisfies the cyclic condition if
\begin{gather*}
		g(\bma,\bmb\cdot \bmc) + g(\bmb,\bmc\cdot \bma) + g(\bmc,\bma\cdot \bmb) = 0,
\end{gather*}
\item the form is totally symmetric if the trilinear form
\begin{gather*}
		c(\bma,\bmb,\bmc):= g(\bma\cdot \bmb,\bmc)
\end{gather*}
is symmetric with respect to all arguments.
\end{itemize}

We are interested in the central extensions of the Lie algebra $\mathscr{L}_\mathbb{A}$. This means that on the direct sum $\widehat{\mathscr{L}}_\mathbb{A} \equiv {\mathscr{L}_\mathbb{A}}\oplus\mathbb{C}$ there is a Lie bracket of the form
\begin{gather*}
		\bc{\bmu\oplus \alpha, \bmv\oplus\beta} := \bc{\bmu,\bmv}\oplus \omega(\bmu,\bmv),
\end{gather*}
where $\omega\colon {\mathscr{L}_\mathbb{A}}\times{\mathscr{L}_\mathbb{A}}\rightarrow \mathbb{C}$ is a $2$-cocycle. This means that $\omega$ is skew-symmetric,
\begin{gather*}
	\omega(\bmu,\bmv) = -\omega(\bmv,\bmu), 	
\end{gather*}
and it satisfies the cyclic condition
\begin{gather*}
	\omega\bra{\bc{\bmu,\bmv},\bmw} + \omega\bra{\bc{\bmv,\bmw},\bmu} + \omega\bra{\bc{\bmw,\bmu},\bmv} = 0.
\end{gather*}

The differential $2$-cocycles, yielding central extensions of the Lie algebras ${\mathscr{L}_\mathbb{A}}$ associated with Novikov algebras,
are generated by appropriate bilinear forms satisfying various algebraic conditions that were originally derived in \cite{BN},
see also~\cite{SS}:
\begin{itemize}\itemsep=0pt
\item A bilinear form $g\colon \mathbb{A}\times\mathbb{A}\rightarrow\mathbb{C}$ defines the first-order $2$-cocycle
\begin{subequations}\label{cocycles}
\begin{gather}\label{cocx}
		\omega(\bmu,\bmv) = \int_\Omega g(\bmu_x,\bmv) \dx
\end{gather}
iff $g$ is symmetric and satisfies the quasi-Frobenius condition.

\item A bilinear form $f\colon \mathbb{A}\times\mathbb{A}\rightarrow\mathbb{C}$ defines the second-order $2$-cocycle
\begin{gather}
		\omega(\bmu,\bmv) = \int_\Omega f(\bmu_{xx},\bmv) \dx
\end{gather}
iff $f$ is skew-symmetric and satisfies the quasi-Frobenius and cyclic conditions.

\item A bilinear form $h\colon \mathbb{A}\times\mathbb{A}\rightarrow\mathbb{C}$ defines the third-order $2$-cocycle
\begin{gather}\label{coc3x}
		\omega(\bmu,\bmv) = \int_\Omega h(\bmu_{xxx},\bmv) \dx
\end{gather}
\end{subequations}
iff $h$ is symmetric and the related trilinear form $c(\bma,\bmb,\bmc)\equiv h(\bma\cdot \bmb,\bmc)$ is totally symmetric.
\end{itemize}
There are no $2$-cocycles of higher order.

Our aim is to complete the above theory with $2$-cocycles by which one can introduce a new independent variable.

\begin{Theorem}\label{centext} A bilinear form $\eta\colon \mathbb{A}\times\mathbb{A}\rightarrow\mathbb{C}$ generates on $\mathscr{L}_{\mathbb{A}}$ the first-order $2$-cocycle:
\begin{gather}\label{cocy}
		\omega(\bmu,\bmv) := \int_\Omega\eta(\bmu_y,\bmv) \dx,
\end{gather}
defined with respect to an $($additional$)$ independent variable $y\in\mathbb{S}^1$, if and only if the form $\eta$ is symmetric, satisfies
the quasi-Frobenius and cyclic conditions.\footnote{These conditions can be directly obtained from the algebraic conditions for the multi-dimensional
Poisson brackets of Dubrovin--Novikov type obtained in \cite{Mokhov1,Mokhov2}, see also~\cite{St2}.}
\end{Theorem}
\begin{proof}Integrating by parts
\begin{gather*}
\omega(\bmu,\bmv) := -\int_{\Omega}\eta(\bmu,\bmv_y) \dx.
\end{gather*}
Then,
\begin{gather*}
\omega\bra{\bc{\bmu,\bmv},\bmw} + \omega\bra{\bc{\bmv,\bmw},\bmu} + \omega\bra{\bc{\bmw,\bmu},\bmv} \\
\quad{} =-\int_{\Omega}\big[ \eta(\bmu_x\cdot \bmv- \bmv_x\cdot \bmu,\bmw_y) + \eta(\bmv_x\cdot \bmw - \bmw_x\cdot \bmv,\bmu_y) +\eta(\bmw_x\cdot \bmu - \bmu_x\cdot \bmw,\bmv_y)\big]\dx\\
\quad{} =\int_{\Omega}\big [ \eta(\bmu_{xy}\cdot \bmv,\bmw) + \eta(\bmu_x\cdot \bmv_y,\bmw) - \eta(\bmv_{xy}\cdot \bmu,\bmw) -\eta(\bmv_x\cdot \bmu_y,\bmw) -\eta(\bmv_x\cdot \bmw,\bmu_y)\\
\quad\quad{}\! - \eta(\bmw\cdot \bmv_x,\bmu_y) - \eta(\bmw\cdot \bmv, \bmu_{xy}) + \eta(\bmw\cdot \bmu_x,\bmv_y) + \eta(\bmu\cdot \bmu,\bmv_{xy}) + \eta(\bmu_x\cdot \bmw,\bmv_y)\big] \dx = 0.
\end{gather*}
Collecting terms with respect to the functionally independent variables, for instance $\{\bmu_{xy},\bmv,\bmw\}$, we obtain the required conditions on the bilinear form~$\eta$.
\end{proof}

Notice that always a linear combination of $2$-cocyles is a $2$-cocycle. This means, that each linearly independent $2$-cocycle of the form \eqref{cocy} can be defined with respect to a different additional independent variable.

In the article \cite{BM1} the three- and less-dimensional Novikov algebras over complex numbers were fully classified. In arbitrary dimension the classification is far from being complete. For instance, in dimension four only transitive Novikov algebras were classified \cite{BM2, BG}. From the point of view of the construction of integrable hierarchies from Novikov algebras the transitive\footnote{A Novikov algebra is called transitive (or right-nilpotent) if every right multiplication is nilpotent.} algebras are not of interest as they result in 'degenerate' systems of evolution equations, see~\cite{SS}.

Recently, in the work \cite{FLS} a classification of non-degenerate two-dimensional Hamiltonian operators of Dubrovin--Novikov type \eqref{mDN} ($D=2$) with small number of components, $n\leqslant 4$, extending previous results from \cite{Mokhov1,Mokhov2} ($n=1,2$), has been completed. Since all $2$-dimensional Hamiltonian operators from that class can be reduced to a linear form the results of \cite{FLS} provides also an implicit classification of associated Novikov algebras for which there exists first-order central extensions \eqref{cocx} and \eqref{cocy} defined by non-degenerate bilinear forms $g$ and $\eta$. Taking advantage of the classification presented in \cite{FLS} we have extracted in Table~\ref{tab2} all the respective unique $4$-dimensional non-transitive Novikov algebras.

For a given Novikov algebra classification of the bilinear forms generating differential central extensions is rather straightforward as it involves only solving
the systems of linear equations, that can be done without much difficulty using any software for symbolic computations. The classification of differential $2$-cocycles, with respect to variable $x$, associated with the low-dimensional Novikov algebras was presented in the work~\cite{SS}. Here, we extend this classification over the case including differential $2$-cocycles of the form~\eqref{cocy}, defined with respect to an additional independent variable. The results are summarized in the following theorem and Tables~\ref{tab1} and~\ref{tab2}. We in fact are only interested with the cases of~\eqref{cocy} defined by means of non-degenerate~$\eta$.

\begin{Theorem}\label{classif} We consider here only Novikov algebras $\mathbb{A}$ that cannot be decomposed into direct sum of lower-dimensional Novikov algebras.
\begin{itemize}\itemsep=0pt
\item In dimension $1$ there is only one relevant Novikov algebra, $\mathbb{A}=\mathbb{C}$, and in this case there is no central extension of the form~\eqref{cocy}.
\item In dimension $2$ there is only one Novikov algebra, $(N6)$ with $\kappa=-2$, for which there exists the central extension~\eqref{cocy} defined by a non-degenerate bilinear form $\eta$.
\item In dimension $3$ there are only two Novikov algebras, $(A7)$ and $(D6)$ both with $\kappa=-2$, for which there exists the central extension \eqref{cocy} defined by a non-degenerate $\eta$.
\item In dimension $4$ there are several Novikov algebras for which there exists the central extension \eqref{cocy} defined by a non-degenerate $\eta$.
\item Up to dimension $4$ there is only one transitive Novikov algebra, $(N3)\otimes (N6)$ with $\kappa=-2$, for which the dimension of the linear space spanned by central extensions \eqref{cocy} is higher than $1$.
\end{itemize}
\end{Theorem}

\begin{Remark}In the work \cite{FLS} it is showed that any two-dimensional irreducible nonconstant Hamiltonian operator of Dubrovin--Novikov type \eqref{mDN} ($D=2$) with three-components ($n=3$) can be reduced to one of two canonical forms. These forms correspond respectively to Novikov algebras $(A7)$ and $(D6)$ both with $\kappa=-2$. By another result from \cite{FLS} any non-degenerated three-dimensional irreducible Hamiltonian operator \eqref{mDN} $(D=3)$ with three-components, which is not transformable to constant coefficients, can be brought to one canonical form. This form corresponds to the transitive algebra $(A7)$ with $\kappa=-2$.
\end{Remark}

\section{Killing form}

The Killing form on ${\mathscr{L}_\mathbb{A}}$ is a symmetric non-degenerate bilinear form $\mcK\colon {\mathscr{L}_\mathbb{A}}\times{\mathscr{L}_\mathbb{A}}\rightarrow\mathbb{C}$ which is $ad$-invariant:
\begin{gather*}
		\mcK\bra{\bc{\bmu,\bmw},\bmv} + \mcK\bra{\bmw,\bc{\bmu,\bmv}} =0 .
\end{gather*}
Assume that $\mcK$ is defined by a bilinear form $\kappa\colon \mathbb{A}\times\mathbb{A}\rightarrow\mathbb{C}$:
\begin{gather}\label{killing}
		\mcK\bra{\bmu,\bmv} := \int_{\Omega}\kappa(\bmu,\bmv) \dx\equiv \dual{\hat{\kappa} \bmu, \bmv},
\end{gather}
where $\hat{\kappa}\colon \mathbb{A}\rightarrow\mathbb{A}^*$. If there exists a Killing form then the dual Lie algebra ${\mathscr{L}_\mathbb{A}}^*$ can be identified
with ${\mathscr{L}_\mathbb{A}}$ and the co-adjoint action can be identified with the adjoint action:
\begin{gather*}
	\ad^*_\bmu\bm{\xi}\equiv \hat{\kappa}\bc{\bmu,\bmw},\qquad \bm{\xi} \equiv \hat{\kappa}\bmw,\qquad \dual{\ad_\bmu^*\bm{\xi},\bmv}:=-\dual{\bm{\xi},\bc{\bmu,\bmv}},
\end{gather*}
where $\bm{\xi}\in{\mathscr{L}_\mathbb{A}}^*$.

\begin{Theorem} A symmetric non-degenerate bilinear form $\kappa\colon \mathbb{A}\times\mathbb{A}\rightarrow\mathbb{C}$ generates the Killing form \eqref{killing} if and only if the form $\kappa$ satisfies the cyclic and quasi-Frobenius conditions.
\end{Theorem}
\begin{proof}We have
\begin{gather*}
\mcK\bra{\bc{\bmu,\bmw},\bmu} + \mcK\bra{\bmw,\bc{\bmu,\bmv}} \\
\quad{} = \int_{\Omega}\big[\kappa(\bmu_x\cdot \bmw - \bmw_x\cdot \bmu, \bmv) + \kappa(\bmw,\bmu_x\cdot \bmv - \bmv_x\cdot \bmu)\big]\dx\\
\quad{}= \int_{\Omega}\big[\kappa(\bmu_x\cdot \bmw, \bmv) + \kappa(\bmw\cdot \bmu_x, \bmv) + \kappa(\bmw\cdot \bmu, \bmv_x) + \kappa(\bmw,\bmu_x\cdot \bmv) - \kappa(\bmw, \bmv_x\cdot \bmu)\big] \dx =0.
\end{gather*}
Collecting terms with respect to functionally independent variables we obtain two conditions
\begin{gather*}
	\kappa(\bma\cdot \bmc, \bmb) + \kappa(\bmc\cdot \bma, \bmb) + \kappa(\bmc,\bma\cdot \bmb) = 0
\end{gather*}
and
\begin{gather*}
	\kappa(\bmc\cdot \bma, \bmb) = \kappa(\bmc, \bmb\cdot \bma).
\end{gather*}
Taking into account the symmetry of the form $\kappa$ we get the conditions from the theorem.
\end{proof}

We see that the algebraic conditions for the Killing forms \eqref{killing} are the same as for the $2$-cocycles~\eqref{cocy} with non-degenerate $\eta$. This means that on a Novikov algebra, on which there exists a `non-degenerate' $2$-cocycle \eqref{cocy}, one can always define the Killing form~\eqref{killing}.

In the following lemma we rewrite the algebraic conditions that must be satisfied by differential central extensions with respect to a given Killing form.

\begin{Lemma}\label{linc} We assume that the Novikov algebra $\mathbb{A}$ is such that there on ${\mathscr{L}_\mathbb{A}}$ exists the associated Killing form~\eqref{killing}.
Then
\begin{itemize} \itemsep=0pt
\item a linear form $\hat{g}\colon \mathbb{A}\rightarrow\mathbb{A}$ generates the two-cocycle of first order,
\begin{gather*}
		\omega(\bmu,\bmv) = \int_\Omega g(\bmu_x,\bmv) \dx \equiv \mcK\big(\hat{g}\bmu_x,\bmv\big),
\end{gather*}
iff the following $($equivalent$)$ relations hold on $\mathbb{A}$
\begin{gather*}
		\hat{g}(\bma\cdot\bmb) = (\hat{g}\bma)\cdot\bmb
\end{gather*}
and
\begin{gather*}
		(\hat{g}\bma)\cdot\bmb - \bma\cdot (\hat{g}\bmb) = (\hat{g}\bmb)\cdot\bma - \bmb\cdot (\hat{g}\bma);
\end{gather*}
\item a linear form $\hat{f}\colon \mathbb{A}\rightarrow\mathbb{A}$ generates the two-cocycle of second order,
\begin{gather*}
		\omega(\bmu,\bmv) = \int_\Omega f(\bmu_{xx},\bmv) \dx \equiv \mcK\big(\hat{f}\bmu_{xx},\bmv\big),
\end{gather*}
iff on $\mathbb{A}$ the following two conditions are satisfied,
\begin{gather*}
		\hat{f}(\bma\cdot\bmb) = \big(\hat{f}\bma\big)\cdot\bmb
\end{gather*}
and
\begin{gather*}
		\hat{f}(\bma\cdot\bmb - \bmb\cdot \bma) = \bma\cdot\big(\hat{f}\bmb\big) + \big(\hat{f}\bmb\big)\cdot\bma;
\end{gather*}
\item a linear form $\hat{h}\colon \mathbb{A}\rightarrow\mathbb{A}$ generates the two-cocycle of third order,
\begin{gather*}
		\omega(\bmu,\bmv) = \int_\Omega h(\bmu_{xxx},\bmv) \dx \equiv \mcK\bra{\hat{h}\bmu_{xxx},\bmv},
\end{gather*}
iff the relations
\begin{gather*}
		\hat{h}(\bma\cdot\bmb) = \hat{h}(\bmb\cdot\bma) = \big(\hat{h}\bma\big)\cdot\bmb = \big(\hat{h}\bmb\big)\cdot\bma
\end{gather*}
and
\begin{gather*}
		\hat{h}(\bma\cdot\bmb) = -\frac{1}{2} \bma\cdot\big(\hat{h}\bmb\big)
\end{gather*}
are valid on $\mathbb{A}$;
\item a linear form $\hat{\eta}\colon \mathbb{A}\rightarrow\mathbb{A}$ generates the two-cocycle of first order with respect to an additional independent variable $y\in\mathbb{S}^1$,
\begin{gather*}
		\omega(\bmu,\bmv) = \int_\Omega \eta(\bmu_{y},\bmv) \dx \equiv \mcK\big(\hat{\eta}\bmu_{y},\bmv\big),
\end{gather*}
iff
\begin{gather*}
		\hat{\eta}(\bma\cdot\bmb) = \big(\hat{\eta}\bma\big)\cdot\bmb
\end{gather*}
and
\begin{gather*}
		(\hat{\eta}\bma)\cdot\bmb = \bma\cdot\big(\hat{\eta}\bmb\big)
\end{gather*}
hold on $\mathbb{A}$.
\end{itemize}
\end{Lemma}

We skip the proof as it is rather straightforward and it does not add anything significant for further considerations.

\section{Bi-Hamiltonian structure}\label{sect}

Consider the centrally extended Lie--Poisson bracket, associated with the Lie algebra ${\mathscr{L}_\mathbb{A}}$, and defined with respect to the Killing form \eqref{killing}:
\begin{gather}\label{poiss1}
		\pobr{\mcH,\mcF}_1(\bmu) = \mcK\bra{\bmu,\bc{\delta_\bmu\mcH, \delta_\bmu\mcF}} + \omega_1\bra{\delta_\bmu\mcH,\delta_\bmu\mcF},
\end{gather}
where $\bmu\in{\mathscr{L}_\mathbb{A}}$ and $\mcH,\mcF\in\mscrF\bra{{\mathscr{L}_\mathbb{A}}}$. Here, $\mscrF\bra{{\mathscr{L}_\mathbb{A}}}$ is the space of functionals on the Lie algebra~${\mathscr{L}_\mathbb{A}}$:
\begin{gather*}
		\mcH(\bmu) = \int_\Omega H[\bmu] \dx,
\end{gather*}
where the densities $H[\bmu]$ are (non-local) smooth functions with respect to all variables\footnote{Here, $\partial_x^{-1}$ is a formal inverse of $\partial_x$ and its definition depends on the analytic assumptions which are not significant from the point of view of the algebraic formalism.}
\begin{gather*}
[\bmu]\equiv \big\{\bmu,\bmu_x,\bmu_y,\bmu_{xx}, \bmu_{xy},\ldots, \partial_x^{-1}\bmu_y,\ldots\big\}.
\end{gather*}
The respective (variational) differentials $\delta_\bmu\mcH$ are defined through the standard formula
\begin{gather}\label{diff}
		\mcK\bra{\delta_\bmu\mcH,\bmv} := \frac{{\rm d}}{{\rm d}\varepsilon}\mcH(\bmu+\varepsilon\bmv)\Bigr |_{\varepsilon = 0},\qquad \bmu,\bmv\in{\mathscr{L}_\mathbb{A}}.	
\end{gather}
The second Poisson bracket compatible with \eqref{poiss1} is
\begin{gather}\label{poiss0}
		\pobr{\mcH,\mcF}_0(\bmu) = \omega_0\bra{\delta_\bmu\mcH,\delta_\bmu\mcF}.
\end{gather}
In the formulas \eqref{poiss1} and \eqref{poiss0} $\omega_1$ and $\omega_0$ are linear compositions of admissible
$2$-cocyles. Naturally, the Poisson brackets \eqref{poiss1} and \eqref{poiss0} are compatible since linear composition of $\omega_1$ and $\omega_0$ is also a $2$-cocycle.

The related Poisson tensors $\mcP_i\colon {\mathscr{L}_\mathbb{A}}\rightarrow{\mathscr{L}_\mathbb{A}}$, defined with respect to the Killing form \eqref{killing},
\begin{gather*}
		\pobr{\mcH,\mcF}_i(\bmu) \equiv \mcK\bra{\mcP_i \delta_\bmu\mcH, \delta_\bmu\mcF},\qquad i=1,2,
\end{gather*}
are
\begin{subequations}\label{pt}
\begin{gather}\label{pt1}
		\mcP_1\bm{\gamma} = \bc{\bmu,\bm{\gamma}}+ \mathcal{I}_1 \bm{\gamma}_x
\end{gather}
and
\begin{gather}\label{pt2}
	\mcP_0\bm{\gamma} =\mathcal{I}_0 \bm{\gamma}_x.
\end{gather}
\end{subequations}
The linear operators $\mcI_i\colon {\mathscr{L}_\mathbb{A}}\rightarrow{\mathscr{L}_\mathbb{A}}$, such that
\begin{gather*}
		\mcK\bra{\mcI_i\bm{\gamma}_x, \bmv} := \omega_i\bra{\bm{\gamma},\bmv},\qquad i=0,1,
\end{gather*}
have the form
\begin{subequations}\label{inop}
\begin{gather}\label{inop0}
		\mcI_0 = \hat{g}_0 + \hat{f}_0\partial_x + \hat{h}_0\partial_x^2 + \sum_{k=1}^N\hat{\eta}_{0,k}\partial_x^{-1}\partial_{y_k},
\end{gather}
and
\begin{gather}\label{inop1}
		\mcI_1 = \hat{g}_1 + \hat{f}_1\partial_x + \hat{h}_1\partial_x^2 + \sum_{k=1}^N\hat{\eta}_{1,k}\partial_x^{-1}\partial_{y_k}.
\end{gather}
\end{subequations}
The forms $\hat{g}_i$, $\hat{f}_i$, $\hat{h}_i$ and $\hat{\eta}_{i,k}$ generate the respective $2$-cocycles, see Lemma~\ref{linc}, and $N$ is the dimension of the
linear space spanned by $2$-cocycles of type~\eqref{cocy} introducing additional spatial variables.

The related (infinite) bi-Hamiltonian chain has the form
\begin{gather}
\bmu_{t_0} = \mcP_0\delta_\bmu\mcH_0\equiv 0,\nonumber\\
\bmu_{t_1} = \mcP_1\delta_\bmu\mcH_0=\mcP_0\delta_\bmu\mcH_1,\nonumber\\
\bmu_{t_2} = \mcP_1\delta_\bmu\mcH_1=\mcP_0\delta_\bmu\mcH_2,\label{hier}\\
\cdots\cdots\cdots\cdots\cdots\cdots\cdots\cdots\cdots , \nonumber
\end{gather}
where $\mcH_i\in\mscrF({\mathscr{L}_\mathbb{A}})$ and $\mcH_0$ is a Casimir of~$\mcP_0$. Such bi-Hamiltonian chain exists since $\mcP_0$ is a~constant Poisson operator, which always has non-trivial kernel. The evolution equations from the hierarchy are of~$(N+2)$ dimension: with respect to an evolution variable $t_i$ and $N+1$ spatial variables $x,y_1,\ldots,y_N$.

Assuming invertibility of $\mcI_0\colon {\mathscr{L}_\mathbb{A}}\rightarrow{\mathscr{L}_\mathbb{A}}$ we can introduce the auxiliary dependent variable~$\bbu$ such that
\begin{gather}\label{inerrel}
\bbu := \mcI_0^{-1} \bmu\quad\iff\quad \bmu = \mcI_0\bbu.
\end{gather}
Here, we understand the invertibility of the operator $\mcI_0$ as of the pseudo-differential operators.

\begin{Theorem}We assume that the Novikov algebra $\mathbb{A}$ possesses the right unity $\bme$. Assuming that $\delta_\bmu\mcH_0 = \bm{e}$, the first two nontrivial evolution equations from the hierarchy~\eqref{hier} takes the form
\begin{gather}	
\bmu_{t_1} = \bmu_x,\nonumber\\
\bmu_{t_2} = \bc{\bmu,\bbu} + \mathcal{I}_1 \bbu_x.\label{inthier}
\end{gather}
The respective Hamiltonians are
\begin{gather}
\mcH_0 = \mcK\bra{\bm{e},\bmu},\nonumber\\
\mcH_1 = \frac{1}{2}\mcK\bra{\bmu,\bbu},\nonumber\\
\mcH_2 = \frac{1}{3}\mcK\bra{\bc{\bbu,\bmu},\partial_x^{-1}\bbu} + \frac{1}{2}\mcK\bra{\mcI_1\bbu,\bbu}.\label{hhier}
\end{gather}
\end{Theorem}
\begin{proof}
The kernel of $\mcP_0$ is spanned by $x$-independent elements of ${\mathscr{L}_\mathbb{A}}$. For simplicity,
we choose that $\delta_\bmu\mcH_0 = \bm{e}$. Thus,
\begin{gather*}
	\bmu_{t_1} = \mcP_1\bm{e} \equiv \bmu_x = \mcP_0\delta_\bmu\mcH_1 \equiv \mcI_0 (\delta_\bmu\mcH_1)_x,
\end{gather*}
and solving for $\delta_\bmu\mcH_1$ we find that $\delta_\bmu\mcH_1 = \bbu$. Hence,
\begin{gather*}
\bmu_{t_2} = \mcP_1\bbu \equiv \bc{\bmu,\bbu} + \mathcal{I}_1 \bbu_x = \mcP_0\delta_\bmu\mcH_2 \equiv \mcI_0 (\delta_\bmu\mcH_2)_x
\end{gather*}
and we have\footnote{Applying \eqref{pt1} to $\delta_\bmu\mcH_2$ one can obtain the third flow, $\bmu_{t_3} = \mcP_1\delta_\bmu\mcH_2$, from the hierarchy \eqref{hier}, but it does not have such a `neat' form as previous flows.}
\begin{gather*}
		\delta_\bmu\mcH_2 = \partial_x^{-1}\mcI_0^{-1}\bc{\bmu,\bbu} + \mcI_0^{-1}\mcI_1\bbu.
\end{gather*}
The Hamiltonians $\mcH_i$ related to the cosymmetries $\delta_\bmu\mcH_i$ can be constructed using the homotopy formula, see~\cite{Blaszak, Olver} and proof of Theorem~1 in~\cite{SS}. Alternatively, one can verify the Hamiltonians~\eqref{hhier} using the formula~\eqref{diff}, but this way is less straightforward as one needs to use the particular properties of the respective $2$-cocycles.
\end{proof}

The evolution equations \eqref{inthier} and the Hamiltonians \eqref{hhier} can be written, in a more explicit form, using the auxiliary variable~$\bbu$,
see Appendix~\ref{app}.

\begin{Remark}
Observe that $\mcH_1$ from \eqref{hhier} is a quadratic Hamiltonian functional:
\begin{gather*}
		\mcH_1 = \frac{1}{2}\mcK\bra{\mcI_0 \bbu, \bbu} \equiv \frac{1}{2}\mcK\big(\bmu, \mcI_0^{-1} \bmu\big).
\end{gather*}
The operator $\mcI_0$ \eqref{inop0}, when it is invertible, can be interpreted as the inertia operator. The variables $\bmu$ and $\bbu$ play the role of the momentum and the velocity. As result, the second non-trivial flow in the hierarchy~\eqref{hhier} can be interpreted as the integrable Euler equation on the centrally extended Lie algebra~${\mathscr{L}_\mathbb{A}}$ corresponding to the geodesic flow on the space with metric defined by means of the inertia operator $\mcI_0$. Compare this with Remark~1 and Appendix~A.3 in the article~\cite{SS}. For more information on the subject we refer the reader to the books \cite{AK,KW}.
\end{Remark}

\section{Examples}

\subsection[Algebra $(N6)$ for $\kappa=-2$]{Algebra $\boldsymbol{(N6)}$ for $\boldsymbol{\kappa=-2}$}

This Novikov algebra is $2$-dimensional, non-abelian and non-associative, it is also non-transitive. The associated Lie algebra ${\mathscr{L}_\mathbb{A}}$ with the Lie bracket~\eqref{lb} is isomorphic to the Lie algebra $\Vect\big(\mathbb{S}^1\big)\ltimes\Vect^*\big(\mathbb{S}^1\big)$, that is the semidirect sum of the algebra of smooth vector fields on the circle $\mathbb{S}^1$, $\Vect\big(\mathbb{S}^1\big)$, with its dual $\Vect^* \big(\mathbb{S}^1\big)$, see~\cite{OR} and~\cite{SeSz}.

The structure matrix and the multiplication of the algebra $\mathbb{A}$ are
\begin{gather*}
		\mathcal{B}=\pmatrx{
 {\bm e}_1 & -2 {\bm e}_2 \\
 {\bm e}_2 & 0 },\qquad \pmatrx{a_1\\ a_2}\cdot\pmatrx{b_1\\ b_2} = \pmatrx{a_1b_1\\ a_2b_1 - 2a_1b_2}.
\end{gather*}
The right unity is $\bme = (1, 0)^{\rm T}$. Taking the advantage of the classification of the form $\eta$ in Table~\ref{tab1}, we can define the Killing form \eqref{killing} by
\begin{gather*}
		\hat{\kappa} = \pmatrx{0 & 1\\ 1 & 0}.
\end{gather*}
For $\bmu\in{\mathscr{L}_\mathbb{A}}$ and the related differential of a functional $\mcH$ we choose
\begin{gather*}
\bmu = \pmatrx{u(x,y)\\ v(x,y)},\qquad \delta_{\bmu}\mcH = \pmatrx{\dfrac{\delta H}{\delta v}\vspace{1mm}\\ \dfrac{\delta H}{\delta u}},
\end{gather*}
as then we get the usual Euclidean duality:
\begin{gather*}
		\mcK\bra{\delta_{\bmu}{\mathcal H},\bmu_t} = \int_{\Omega} \left(\frac{\delta H}{\delta u}u_t + \frac{\delta H}{\delta v}v_t\right)\dx,
\end{gather*}
where $\Omega = \mathbb{S}^1\times\mathbb{S}^1$ and $\dx\equiv {\rm d}x{\rm d}y$.

Using the classification of central extensions in Table~\ref{tab1} the operators~\eqref{inop} are
\begin{gather*}
	\mcI_0 = \pmatrx{g_1 & 0\\ g_2 & g_1} + \pmatrx{0 & 0\\ h & 0}\partial_x^2 + \pmatrx{\eta & 0\\ 0 & \eta}\partial_x^{-1}\partial_y
\end{gather*}
and
\begin{gather*}
	\mcI_1 =
\pmatrx{0 & 0\\ \beta & 0}\partial_x^2 + \pmatrx{\gamma & 0\\ 0 & \gamma}\partial_x^{-1}\partial_y.
\end{gather*}
In this case the first-order cocycle is trivial, this means that in the Poisson tensor~\eqref{pt1} $\hat{g}_1$ can be obtained/removed through a shift of the dependent variables. So, we do not consider~$\hat{g}_1$ in~$\mcI_1$. The relation~\eqref{inerrel} between the momentum variable $\bmu$ and the velocity variable $\bbu= (\bar{u}(x,y),\bar{v}(x,y))^{\rm T}$ is
\begin{gather*}
u = g_1 \bar{u} + \eta \partial_x^{-1}\bar{u}_y,\\
v = g_1 \bar{v} + \eta \partial_x^{-1}\bar{v}_y + g_2\bar{u} + h\bar{u}_{xx}.
\end{gather*}
The second flow from the hierarchy \eqref{hier} takes the form
\begin{gather*}
u_{t_2} = u_x\bar{u}-u\bar{u}_x + \gamma\bar{u}_y,\\
v_{t_2} = \bar{u}v_x-u\bar{v}_x -2u_x\bar{v} + 2\bar{u}_xv + \gamma\bar{v}_y + \beta\bar{u}_{xxx}.
\end{gather*}
The related Poisson tensors \eqref{pt} are
\begin{gather*}
	\mcP_0 = \pmatrx{0 & g_1\partial_x + \eta\partial_y\\ g_1\partial_x + \eta\partial_y & g_2\partial_x + h\partial_x^3}
\end{gather*}
and
\begin{gather*}
	\mcP_1 = \pmatrx{0 & u_x-u\partial_x + \alpha_1\partial_x + \gamma\partial_y\\
	-2u_x-u\partial_x + \alpha_1\partial_x + \gamma\partial_y & v_x+2v\partial_x + \alpha_2\partial_x + \beta\partial_x^3}.
\end{gather*}
Notice that the above Poisson tensors are defined with respect to the standard duality pairing. The Hamiltonian functionals \eqref{hhier} are
\begin{gather*}
 \mathcal{H}_0 = \int_{\Omega}v \dx,\\
 \mathcal{H}_1 = \frac{1}{2}\int_{\Omega}\bra{\bar{u}v + u\bar{v}} \dx,\\
\mathcal{H}_2 = \int_{\Omega}\left[\frac{1}{3}\bra{2 u_x\bar{v} - \bar{u} v_x - 2 \bar{u}_x v + u \bar{v}_x}\partial_x^{-1}\bar{u} + \frac{1}{3}\bra{u \bar{u}_x - u_x \bar{u}}\partial_x^{-1}\bar{v} \right.\\
\left.\hphantom{\mathcal{H}_2 =}{} + \frac{1}{2}\big(\gamma \bar{u}\partial_x^{-1}\bar{v}_y + \beta \bar{u}\bar{u}_{xx} + \gamma \bar{v} \partial_x^{-1}\bar{u}_y \big)\right]\dx.
\end{gather*}
The particular cases $\eta_1 = \eta_2=h=0$ and $g_1 = g_2=\beta=0$ of the above bi-Hamiltonian structure were considered in the work \cite{SeSz},
see the two examples therein.

\subsection[Algebra $(D6)$ for $\kappa=-2$]{Algebra $\boldsymbol{(D6)}$ for $\boldsymbol{\kappa=-2}$}

This Novikov algebra can be considered as the $3$-dimensional extension of the algebra from the previous section as the subalgebra
spanned by $\{\bme_1,\bme_2\}$ is identical to the algebra $(N6)$ ($\kappa=-2$).

We proceed as before. The structure matrix and multiplication are
\begin{gather*}
		\mathcal{B} = \pmatrx{
 \bme_1 & -2 \bme_2 & -\frac{1}{2}\bme_3 \\
 \bme_2 & 0 & 0 \\
 \bme_3 & 0 & -\frac{1}{2}\bme_2 },\qquad \pmatrx{a_1\\ a_2\\ a_3}\cdot\pmatrx{b_1\\ b_2\\ b_3} =
\pmatrx{a_1b_1\\ a_2b_1 - 2a_1b_2-\frac{1}{2}a_3b_3\\ a_3b_1-\frac{1}{2}a_1b_3}.
\end{gather*}
The right unity is $\bme = (1, 0,0)^{\rm T}$. The Killing form \eqref{killing} is defined by
\begin{gather*}
		\hat{\kappa} := \pmatrx{0 & 1 & 0\\ 1 & 0 & 0\\ 0 & 0 & 1}.
\end{gather*}
For $\bmu\in{\mathscr{L}_\mathbb{A}}$ and related differentials we choose
\begin{gather*}
\bmu = \pmatrx{u(x,y)\\ v(x,y)\\ w(x,y)},\qquad \delta_{\bmu}{\mathcal H} = \pmatrx{\dfrac{\delta H}{\delta v}\vspace{2mm}\\ \dfrac{\delta H}{\delta u}\vspace{2mm}\\ \dfrac{\delta H}{\delta w}},
\end{gather*}
as then we have
\begin{gather*}
		\mcK\bra{\delta_{\bmu}{\mathcal H},\bmu_t} = \int_{\Omega} \left(\frac{\delta H}{\delta u}u_t + \frac{\delta H}{\delta v}v_t+ \frac{\delta H}{\delta w}w_t\right)\dx,
\end{gather*}
where $\Omega = \mathbb{S}^1\times\mathbb{S}^1$ and $\dx\equiv {\rm d}x{\rm d}y$. The operators \eqref{inop} are
\begin{gather*}
	\mcI_0 = \pmatrx{ g_1 & 0 & 0 \\ g_2 & g_1 & g_3 \\ g_3 & 0 & g_1}
+ \pmatrx{ 0 & 0 & 0 \\ h & 0 & 0 \\ 0 & 0 & 0}\partial_x^2 + \pmatrx{\eta & 0 & 0\\ 0 & \eta & 0\\ 0 & 0 & \eta}\partial_x^{-1}\partial_y
\end{gather*}
and
\begin{gather*}
	\mcI_1 =
\pmatrx{ 0 & 0 & 0 \\ \beta & 0 & 0 \\ 0 & 0 & 0}\partial_x^2 + \pmatrx{\gamma & 0 & 0\\ 0 & \gamma & 0\\ 0 & 0 & \gamma}\partial_x^{-1}\partial_y.
\end{gather*}
The relation \eqref{inerrel} for the auxiliary variables $\bbu = (\bar{u}(x,y), \bar{v}(x,y), \bar{w}(x,y))^{\rm T}$ is
\begin{gather*}
u = g_1 \bar{u} + \eta \partial_x^{-1}\bar{u}_y,\\
v = g_2\bar{u} + g_1 \bar{v}+ g_3\bar{w} + \eta \partial_x^{-1}\bar{v}_y + h\bar{u}_{xx},\\
w = g_3\bar{u} + g_1\bar{w} + \eta \partial_x^{-1}\bar{w}_y .
\end{gather*}
Then, we find the Poisson tensors~\eqref{pt}
\begin{gather*}
	\mcP_0 = \pmatrx{0 & g_1\partial_x + \eta\partial_y & 0\\ g_1\partial_x + \eta\partial_y & g_2\partial_x + h\partial_x^3 & g_3\partial_x\\ 0 & g_3\partial_x & g_1\partial_x + \eta\partial_y}
\end{gather*}
and
\begin{gather*}
	\mcP_1 = \pmatrx{0 & u_x-u\partial_x + \gamma\partial_y & 0\\
	-2u_x-u\partial_x + \gamma\partial_y & v_x+2v\partial_x + \beta\partial_x^3 & -\frac{1}{2}w_x +\frac{1}{2}w\partial_x \\
 0 & w_x + \frac{1}{2}w\partial_x & -\frac{1}{2}u_x-u\partial_x + \gamma\partial_y },
\end{gather*}
the second system from the hierarchy \eqref{inthier}
\begin{gather*}
u_{t_2} = u_x\bar{u} - u\bar{u}_x + \gamma \bar{u}_y,\\
v_{t_2} = -2 u_x\bar{v} + \bar{u}v_x + 2 u_x v - u \bar{v}_x - \frac{1}{2} w_x\bar{w} + \frac{1}{2} w \bar{w}_x+ \gamma \bar{v}_y + \beta \bar{u}_{xxx},\\
w_{t_2} = -\frac{1}{2} u_x \bar{w} + \bar{u} w_x + \frac{1}{2} \bar{u}_x w - u \bar{w}_x+ \gamma \bar{w}_y,
\end{gather*}
and the related Hamiltonians \eqref{hhier}
\begin{gather*}
 \mathcal{H}_0 = \int_{\Omega}v \dx,\\
 \mathcal{H}_1 = \frac{1}{2}\int_{\Omega}\bra{\bar{u}v + u \bar{v} + w \bar{w}} \dx,\\
\mathcal{H}_2 = \int_{\Omega}\left[\frac{1}{3}\left(2 u_x\bar{v} - \bar{u} v_x - 2 \bar{u}_x v + u \bar{v}_x + \frac{1}{2} w_x\bar{w} - \frac{1}{2} w \bar{w}_x\right)\partial_x^{-1}\bar{u}\right.\\
\hphantom{\mathcal{H}_2 =}{} + \frac{1}{3}\bra{u \bar{u}_x - u_x \bar{u}}\partial_x^{-1}\bar{v} +\frac{1}{3}\left(\frac{1}{2} u_x\bar{w} - \bar{u} w_x - \frac{1}{2} \bar{u}_x w + u \bar{w}_x\right)\partial_x^{-1}\bar{w}\\
\left.\hphantom{\mathcal{H}_2 =}{} + \frac{1}{2}\big(\gamma \bar{u}\partial_x^{-1}\bar{v}_y + \beta \bar{u}\bar{u}_{xx} + \gamma \bar{v} \partial_x^{-1}\bar{u}_y + \gamma \bar{w} \partial_x^{-1}\bar{w}_y\big)\right]\dx.
\end{gather*}
Taking $w=\bar{w}=0$ and $g_3=0$ the above evolution equations and the related bi-Hamiltonian structure reduce to the case from the previous section.

\subsection[Algebra $(N3)\otimes(N6)$ for $\kappa=-2$]{Algebra $\boldsymbol{(N3)\otimes(N6)}$ for $\boldsymbol{\kappa=-2}$}\label{dim4}

This Novikov algebra is the tensor product of the associative and abelian algebra $(N3)$ with the Novikov algebra $(N6)$ for $\kappa=-2$. It is non-abelian, non-associative, and also non-transitive. Up to dimension $4$ this is the only non-transitive Novikov algebra that allows for the construction of the $(3+1)$-dimensional bi-Hamiltonian hierarchy~\eqref{hier}.

The structure matrix and algebra multiplication are:\
\begin{gather*}
		\mathcal{B} = \pmatrx{
 \bme_1 & -2 \bme_2 & \bme_3 & -2\bme_4\\
 \bme_2 & 0 & \bme_4 & 0 \\
 \bme_3 & -2\bme_4 & 0 & 0\\
\bme_4 & 0 & 0 & 0},\qquad \pmatrx{a_1\\ a_2\\ a_3\\ a_4}\cdot\pmatrx{b_1\\ b_2\\ b_3\\ b_4} =
\pmatrx{a_1b_1\\ a_2b_1 - 2a_1b_2\\ a_3b_1 + a_1b_3\\ a_4b_1 - 2a_3b_2 + a_2b_3 - 2a_1 b_4}.
\end{gather*}
The right unity is $\bme = (1, 0, 0, 0)^{\rm T}$.

As the Killing form \eqref{killing} we can take
\begin{gather*}
		\hat{\kappa} := \pmatrx{0 & 0 & 0 & 1\\ 0 & 0 & 1 & 0\\ 0 & 1 & 0 & 0\\ 1 & 0 & 0 & 0},
\end{gather*}
thus $\bmu = \bra{u(x,y,z), v(x,y,z), w(x,y,z), s(x,y,z)}^{\rm T}$ and $\delta_{\bmu}{\mathcal H} = \big(\frac{\delta H}{\delta s}, \frac{\delta H}{\delta w},
\frac{\delta H}{\delta v}, \frac{\delta H}{\delta u}\big)^{\rm T}$ so that
\begin{gather*}
		\mcK\bra{\delta_{\bmu}{\mathcal H},\bmu_t} = \int_{\Omega} \left(\frac{\delta H}{\delta u}u_t + \frac{\delta H}{\delta v}v_t+ \frac{\delta H}{\delta w}w_t
+ \frac{\delta H}{\delta s}s_t\right)\dx,
\end{gather*}
where $\Omega = \mathbb{S}^1\times\mathbb{S}^1\times\mathbb{S}^1$ and $\dx\equiv {\rm d}x{\rm d}y{\rm d}z$.

In this case there are two linearly independent central extensions of the form~\eqref{cocy}, see Table~\ref{tab2}. Hence, we introduce two additional
independent variables $y$ and $z$. To simplify this example, we further consider only the central extension of the type \eqref{coc3x} and \eqref{cocy}.

As result, the operators \eqref{inop} are
\begin{gather*}
\mcI_0 = \pmatrx{0 & 0 & 0 & 0 \\ h_1 & 0 & 0 & 0 \\ 0 & 0 & 0 & 0 \\ h_2 & 0 & h_1 & 0}\partial_x^{2}
+ \pmatrx{\eta _1 & 0 & 0 & 0 \\ 0 & \eta _1 & 0 & 0 \\ 0 & 0 & \eta _1 & 0 \\ 0 & 0 & 0 & \eta _1}\partial_x^{-1}\partial_y
+ \pmatrx{0 & 0 & 0 & 0 \\ 0 & 0 & 0 & 0 \\ \eta _2 & 0 & 0 & 0 \\ 0 & \eta _2 & 0 & 0}\partial_x^{-1}\partial_z
\end{gather*}
and
\begin{gather*}
\mcI_1 = \pmatrx{0 & 0 & 0 & 0 \\ \beta_1 & 0 & 0 & 0 \\ 0 & 0 & 0 & 0 \\ \beta_2 & 0 & \beta_1 & 0}\partial_x^{2}
+ \pmatrx{\gamma _1 & 0 & 0 & 0 \\ 0 & \gamma _1 & 0 & 0 \\ 0 & 0 & \gamma _1 & 0 \\ 0 & 0 & 0 & \gamma _1}\partial_x^{-1}\partial_y
+ \pmatrx{0 & 0 & 0 & 0 \\ 0 & 0 & 0 & 0 \\ \gamma _2 & 0 & 0 & 0 \\ 0 & \gamma _2 & 0 & 0}\partial_x^{-1}\partial_z.
\end{gather*}
Next, for $\bbu = (\bar{u}(x,y), \bar{v}(x,y), \bar{w}(x,y), \bar{s}(x,y,z))^{\rm T}$ the relation \eqref{inerrel} is
\begin{gather*}
u = \eta_1\partial_x^{-1}\bar{u}_y,\\
v = \eta_1\partial_x^{-1}\bar{v}_y + h_1\bar{u}_{xx},\\
w = \eta_1\partial_x^{-1}\bar{w}_y + \eta_2\partial_x^{-1}\bar{u}_z,\\
s = \eta_1\partial_x^{-1}\bar{s}_y + \eta_2\partial_x^{-1}\bar{v}_z+ h_1\bar{w}_{xx}+ h_2\bar{u}_{xx}.
\end{gather*}
In this case, the three-dimensional compatible Poisson tensors \eqref{pt} are
\begin{gather*}
		\mcP_0 = \pmatrx{0 & 0 & 0 & \eta_1\partial_y \\ 0 & 0 & \eta_1\partial_y & h_1\partial_x^3 \\ 0 & \eta_1\partial_y & 0 & \eta_2\partial_z \\ \eta_1\partial_y & h_1\partial_x^3 & \eta_2\partial_z & h_2\partial_x^3}
\end{gather*}
and
\begin{gather*}
	\mcP_1 = \pmatrx{0 & 0 & 0 & u_x-u\partial_x + \gamma_1\partial_y\\ 0 & 0 & -2u_x-u\partial_x + \gamma_1\partial_y & v_x+2v\partial_x + \beta_1\partial_x^3\\
0 & u_x-u\partial_x + \gamma_1\partial_y& 0 & w_x-w\partial_x + \gamma_1\partial_z\\ -2u_x-u\partial_x + \gamma_2\partial_y\!& v_x+2v\partial_x + \beta_1\partial_x^3 & -2w_x-w\partial_x + \gamma_2\partial_z & s_x+2s\partial_x + \beta_2\partial_x^3}.	
\end{gather*}
The second nontrivial system from the hierarchy \eqref{inthier} is
\begin{gather*}
u_{t_2} = u_x\bar{u} - u \bar{u}_x + \gamma_1\bar{u}_y,\\
v_{t_2} = -2 u_x\bar{v} + \bar{u}v_x + 2 \bar{u}_xv - u\bar{v}_x + \gamma_1\bar{v}_y + \beta_1 \bar{u}_{xxx},\\
w_{t_2} = u_x\bar{w} + \bar{u}w_x - \bar{u}_x w - u\bar{w}_x + \gamma_1\bar{w}_y + \gamma_2\bar{u}_z,\\
s_{t_2} = \bar{u}s_x - 2 u_x\bar{s} + v_x\bar{w} - 2 \bar{v}w_x - u\bar{s}_x + 2 \bar{u}_x s - \bar{v}_x w + 2 v\bar{w}_x\\
\hphantom{s_{t_2} =}{} + \gamma_1\bar{s}_y + \gamma_2\bar{v}_z + \beta_1 \bar{w}_{xxx} + \beta_2 \bar{u}_{xxx}.
\end{gather*}
Notice that this is a $(3+1)$-dimensional (dispersive) integrable bi-Hamiltonian evolution system and it can be explicitly written either by means of the components of the momentum variable~$\bmu$ or the velocity variable~$\bbu$. The respective Hamiltonians \eqref{hhier} are
\begin{gather*}
 \mathcal{H}_0 = \int_{\Omega}s\dx,\\
 \mathcal{H}_1 = \frac{1}{2}\int_{\Omega}\bra{u \bar{s} + \bar{u} s + \bar{v} w + v \bar{w}} \dx,
\end{gather*}
and
\begin{gather*}
 \mathcal{H}_2 = \int_{\Omega}\left[\frac{1}{3}\bra{-\bar{u}s_x + 2 u_x\bar{s} + u \bar{s}_x - 2 \bar{u}_x s - v_x\bar{w} + 2 \bar{v} w_x + \bar{v}_x w - 2 v \bar{w}_x}\partial_x^{-1}\bar{u}\right.\\
\hphantom{\mathcal{H}_2 =}{} - \frac{1}{3}\bra{u_x\bar{w} + \bar{u} w_x - \bar{u}_x w - u \bar{w}_x}\partial_x^{-1}\bar{v}
+ \frac{1}{3}\bra{2 u_x\bar{v} - \bar{u} v_x - 2 \bar{u}_x v + u \bar{v}_x}\partial_x^{-1}\bar{w}\\
\hphantom{\mathcal{H}_2 =}{} + \frac{1}{3}\bra{u \bar{u}_x - u_x\bar{u}}\partial_x^{-1}\bar{s}
+\frac{1}{2}\big(\gamma_1\bar{s}\partial_x^{-1}\bar{u}_y + \gamma_1\bar{w}\partial_x^{-1}\bar{v}_y
 + \gamma_1\bar{v}\partial_x^{-1}\bar{w}_y + \gamma_2\bar{v}\partial_x^{-1}\bar{u}_z\big)\\
\left.\hphantom{\mathcal{H}_2 =}{} +\frac{1}{2}\big(\gamma_1\bar{u}\partial_x^{-1}\bar{s}_y + \gamma_2\bar{u}\partial_x^{-1}\bar{v}_z
+\beta_1\bar{u}_{xx}\bar{w} + \beta_1\bar{u}\bar{w}_{xx} + \beta_2\bar{u}\bar{u}_{xx}\big)\right]\dx.
\end{gather*}
Notice that in the example from this section one can reduce dimension taking $y=z$ or $y=x$ or $z=x$.

\appendix

\section{Explicit form of the Bi-Hamiltonian hierarchy}\label{app}

In the variable $\bbu$ the bi-Hamiltonian chain \eqref{hier} takes the form
\begin{gather*}
\bbu_{t_0} = \bar{\mcP}_0\delta_\bbu\mcH_0\equiv 0,\\
\bbu_{t_1} = \bar{\mcP}_1\delta_\bbu\mcH_0=\bar{\mcP}_0\delta_\bbu\mcH_1,\\
\bbu_{t_2} = \bar{\mcP}_1\delta_\bbu\mcH_1=\bar{\mcP}_0\delta_\bbu\mcH_2,\\
\cdots\cdots\cdots\cdots\cdots\cdots\cdots\cdots\cdots,
\end{gather*}
where $\bar{\mcP}_i = \mcI_0^{-1}\mcP_i\mcI_0^{-1}$ and $\delta_\bbu\mcH = \mcI_0\delta_\bmu\mcH$. Then,
\begin{gather}
		\bbu_{t_1} = \bbu_x,\qquad
		\mathcal{I}_0 \bbu_{t_2} = \bc{\mathcal{I}_0 \bbu,\bbu} + \mathcal{I}_1 \bbu_x\label{uhier}
\end{gather}
and
\begin{gather*}
		\mcH_0 = \mcK\bra{\bm{c},\mcI_0\bbu},\qquad
		 \mcH_1 = \frac{1}{2}\mcK\bra{\mcI_0\bbu,\bbu},\qquad
		 \mcH_2 = \frac{1}{3}\mcK\big(\bc{\bbu,\mcI_0\bbu},\partial_x^{-1}\bbu\big)+ \frac{1}{2}\mcK\bra{\mcI_1\bbu,\bbu}.
\end{gather*}

\begin{Proposition} The explicit form, in the variable $\bbu$, of the second flow from \eqref{uhier} is
\begin{gather*}
	\hat{g}_0 \bbu_{t_2} + \hat{f}_0 \bbu_{xt_2} + \hat{h}_0 \bbu_{xxt_2} + \sum_{k=1}^N\hat{\eta}_{0,k} \partial_x^{-1}\bbu_{y_k t_2} \\
	=\left(\frac{1}{2}\hat{g}_0\bbu\cdot\bbu - \frac{1}{2}\bbu\cdot\hat{g}_0\bbu + \hat{f}_0\bbu_x\cdot\bbu +
	\hat{h}_0\bbu_{xx}\cdot\bbu + \frac{1}{2}\hat{h}_0\bbu_x\cdot\bbu_x
	 -\sum_{k=1}^N\hat{\eta}_{0,k}\bbu\cdot\partial_x^{-1}\bbu_{y_k} \right)_x \\
\quad{} + \sum_{k=1}^N\hat{\eta}_{0,k}(\bbu\cdot\bbu)_{y_k}
+ \hat{g}_1\bbu_x + \hat{f}_1\bbu_{xx} + \hat{h}_1\bbu_{xxx} + \sum_{k=1}^N\hat{\eta}_{1,k}\bbu_{y_k},
\end{gather*}
and the explicit form of Hamiltonians is
\begin{gather*}
	\mcH_0 = \mcK\bra{\hat{g}_0\bmc,\bbu},\\
	\mcH_1 = \frac{1}{2}\mcK\bra{\hat{g}_0\bbu,\bbu} + \frac{1}{2}\mcK\bra{\hat{f}_0\bbu_x,\bbu} + \frac{1}{2}\mcK\bra{\hat{h}_0\bbu_{xx},\bbu}
	+\frac{1}{2}\sum_{k=1}^N\mcK\bra{\hat{\eta}_{0,k}\partial_x^{-1}\bbu_{y_k},\bbu},\\
	\mcH_2 = \frac{1}{2}\mcK\bra{\hat{g}_0\bbu,\bbu\cdot \bbu} + \frac{1}{3}\mcK\bra{\hat{f}_0\bbu_x,\bbu\cdot \bbu}
	+ \frac{1}{3}\mcK\bra{\hat{h}_0\bbu_{xx} ,\bbu\cdot \bbu} + \frac{1}{6}\mcK\bra{\hat{h}_0\bbu_x,\bbu\cdot \bbu_x}\\
\hphantom{\mcH_2 =}{}+ \frac{1}{3}\sum_{k=1}^N\mcK\bra{\hat{\eta}_{0,k}\bbu,\partial_x^{-1}\bbu_{y_k}\cdot \bbu} - \frac{1}{3}\sum_{k=1}^N \mcK\bra{\hat{\eta}_{0,k}\bbu,\bbu\cdot\partial_x^{-1}\bbu_{y_k}}\\
\hphantom{\mcH_2 =}{} + \frac{1}{2}\mcK\bra{\hat{g}_1\bbu,\bbu} + \frac{1}{2}\mcK\bra{\hat{f}_1\bbu_x,\bbu} + \frac{1}{2}\mcK\bra{\hat{h}_1\bbu_{xx},\bbu} + \frac{1}{2}\sum_{k=1}^N\mcK\bra{\hat{\eta}_{1,k}\partial_x^{-1}\bbu_{y_k},\bbu}.
\end{gather*}
\end{Proposition}
\begin{proof}We will skip the details. The proposition is a consequence of the following relations
\begin{gather*}
\bc{\hat{g}_0\bbu,\bbu} = (\hat{g}_0\bbu_x)\cdot\bbu - \bbu_x\cdot(\hat{g}_0\bbu) =
\frac{1}{2}\bra{\hat{g}_0\bbu\cdot\bbu - \bbu\cdot\hat{g}_0\bbu}_x,\\
\bc{\hat{f}_0\bbu_x,\bbu} = (\hat{f}_0\bbu_{xx})\cdot\bbu - \bbu_x\cdot(\hat{f}_0\bbu_x) =
\hat{f}_0(\bbu_x\cdot\bbu)_x,\\
\bc{\hat{h}_0\bbu_{xx},\bbu} = (\hat{h}_0\bbu_{xxx})\cdot\bbu - \bbu_{x}\cdot(\hat{h}_0\bbu_{xx}) =
\hat{h}_0\left(\bbu_{xx}\cdot\bbu+\frac{1}{2}\bbu_x\cdot\bbu_x\right)_x,\\
\bc{\hat{\eta}_{0,k} \partial_x^{-1}\bbu_{y_k},\bbu} = (\hat{\eta}_{0,k}\bbu_{y_k})\cdot\bbu - \bbu_x\cdot\big(\hat{\eta}_{0,k}\partial_x^{-1}\bbu_{y_k}\big) =
\hat{\eta}_{0,k}(\bbu\cdot\bbu)_{y_k} - \hat{\eta}_{0,k}\big(\bbu\cdot\partial_x^{-1}\bbu_{y_k}\big)_x,
\end{gather*}
which can be proven using all the properties of the linear forms generating respective $2$-cocycles obtained in Lemma~\ref{linc}.
\end{proof}

The above proposition is the higher-dimensional extension of Theorem~1 from the article \cite{SS}, see also the equation~(26) therein.

{\arraycolsep=3pt
\setlength{\tabcolsep}{-2pt}
\begin{landscape}
\begin{table}[h!]
\centering
 \caption{The classification of the bilinear forms generating respective differential $2$-cocycles~\eqref{cocycles} and~\eqref{cocy} for all non-decomposable $1$ and $2$-dimensional Novikov algebras and the significant $3$-dimensional cases corresponding to Theorem~\ref{classif}. We use for the Novikov algebras the same names as in~\cite{BM1}.}\label{tab1}
 \vspace{1mm}

\begin{tabular}{ c c c c c c c }
\hline type & $\mathcal{B}$ & $g$ & $f$ & $h$ & $\eta$ & comments \\\hline
$\mathbb{C}$ & $e_1$& $g_{11}$ & $0$ & $h_{11}$ & 0 & $\begin{array}{@{}c@{}}\textf{abelian}\\[-1ex] \textf{associative}\end{array}$\tsep{2pt}\bsep{1pt}\\
$(T2)$ & $
\pmatrx{
 \bme_2 & 0 \\
 0 & 0}
$ & $
\pmatrx{
 g_{1 1} & g_{1 2} \\
 g_{1 2} & 0}
$ & $
\pmatrx{
 0 & 0 \\
 0 & 0}
$ & $
\pmatrx{
 h_{1 1} & h_{1 2} \\
 h_{1 2} & 0}
$ & $
\pmatrx{
 \eta _{11} & 0 \\
 0 & 0}
$ & $
\begin{array}{@{}c@{}}\textf{abelian}\\[-1ex] \textf{associative}\\[-1ex]\textf{transitive}\end{array}
$\bsep{10pt}\\
$(T3)$ & $
\pmatrx{
 0 & 0 \\
 -\bme_1 & 0 }
$ & $
\pmatrx{
 0 & g_{1 2} \\
 g_{1 2} & g_{2 2} }
$ & $
\pmatrx{
 0 & 0 \\
 0 & 0 }
$ & $
\pmatrx{
 0 & 0 \\
 0 & h_{2 2} }
$ &
$\pmatrx{0 & 0\\
0 & \eta_{22} }
$ &
$\textf{transitive}$\bsep{10pt}\\
$(N3)$ & $
\pmatrx{
 \bme_1 & \bme_2 \\
 \bme_2 & 0 }
$ & $
\pmatrx{
 g_{1 1} & g_{1 2} \\
 g_{1 2} & 0 }
$ & $
\pmatrx{
 0 & 0 \\
 0 & 0 }
$ & $
\pmatrx{
 h_{1 1} & h_{1 2} \\
 h_{1 2} & 0 }
$ & $
\pmatrx{
 0 & 0 \\
 0 & 0 }
$ &
 $\begin{array}{@{}c@{}}\textf{abelian}\\[-1ex] \textf{associative}\end{array}$\bsep{10pt}\\
$(N4)$ & $
\pmatrx{
 0 & \bme_1 \\
 0 & \bme_2 }
$ & $
\pmatrx{
 g_{1 1} & g_{1 2} \\
 g_{1 2} & g_{2 2} }
$ & $
\pmatrx{
 0 & f_{1 2} \\
 -f_{1 2} & 0 }
$ & $
\pmatrx{
 0 & 0 \\
 0 & h_{2 2} }
$ & $
\pmatrx{
 0 & 0 \\
 0 & 0 }
 $ & $\textf{associative}$\bsep{10pt}\\
$(N5)$ & $
\pmatrx{
 0 & \bme_1 \\
 0 & \bme_1+\bme_2 \\
}
$ & $
\pmatrx{
 0 & g_{1 2} \\
 g_{1 2} & g_{2 2} \\
}
$ & $
\pmatrx{
 0 & 0 \\
 0 & 0 \\
}
$ & $
\pmatrx{
 0 & 0 \\
 0 & h_{2 2} \\
}
$ & $
\pmatrx{
 0 & 0 \\
 0 & 0 }
$ & \bsep{10pt}\\
$\begin{array}{@{}c@{}}(N6)\\[-.5ex] \kappa\neq 0,1,-2\end{array}$ &
$
\pmatrx{
 \bme_1 & \kappa \bme_2 \\
 \bme_2 & 0
}
$ & $
\pmatrx{
 g_{11} & g_{12} \\
 g_{12} & 0
}
$ & $
\pmatrx{
 0 & 0 \\
 0 & 0 \\
}
$ & $
\pmatrx{
 h_{11} & 0 \\
 0 & 0 \\
}
$ & $
\pmatrx{
 0 & 0 \\
 0 & 0 }
 $ &
 \bsep{10pt}\\
$\begin{array}{@{}c@{}}(N6)\\[-.5ex] \kappa=-2\end{array}$ &
$
\pmatrx{
 \bme_1 & -2\bme_2 \\
 \bme_2 & 0 }
$ & $
\pmatrx{
 g_{11} & g_{1 2} \\
 g_{1 2} & 0 }
$ & $
\pmatrx{
 0 & 0 \\
 0 & 0 }
$ & $
\pmatrx{
h_{11} & 0 \\
 0 & 0 }
$ & $
\pmatrx{0 & \eta_{12}\\
\eta_{12} & 0 }
$ & $\det \eta \neq 0$ \bsep{10pt}\\
$\begin{array}{@{}c@{}}(A7)\\[-.5ex] \kappa\neq 1,-2\end{array}$ &
$
\pmatrx{
 \bme_2 & \kappa \bme_3 & 0 \\
 \bme_3 & 0 & 0 \\
 0 & 0 & 0
 }
$ & $
\pmatrx{
 g_{11} & g_{12} & g_{13} \\
 g_{12} & g_{13} & 0 \\
 g_{13} & 0 & 0
}
$ & $
\pmatrx{
0 & 0 & 0\\
0 & 0 & 0\\
0 & 0 & 0}
$ & $
\pmatrx{
 h_{11} & h_{12} & 0 \\
 h_{12} & 0 & 0 \\
 0 & 0 & 0
 }
$ & $
\pmatrx{
 \eta _{11} & 0 & 0\\
 0 & 0 & 0 \\
 0 & 0 & 0
 }
 $ &
 $\textf{transitive}$\bsep{10pt}\\
$\begin{array}{@{}c@{}}(A7)\\[-.5ex] \kappa=-2\end{array}$ &
$
\pmatrx{
 \bme_2 & -2 \bme_3 & 0 \\
 \bme_3 & 0 & 0 \\
 0 & 0 & 0
 }
$ & $
\pmatrx{
 g_{11} & g_{12} & g_{13} \\
 g_{12} & g_{13} & 0 \\
 g_{13} & 0 & 0
}
$ & $
\pmatrx{
0 & 0 & 0\\
0 & 0 & 0\\
0 & 0 & 0}
$ & $
\pmatrx{
 h_{11} & h_{12} & 0 \\
 h_{12} & 0 & 0 \\
 0 & 0 & 0
 }
$ & $
\pmatrx{
 \eta _{11} & 0 & \eta _{13} \\
 0 & \eta _{13} & 0 \\
 \eta _{13} & 0 & 0
 }
 $ &
 $\begin{array}{@{}c@{}}\textf{transitive}\\[-.5ex] \det \eta \neq 0\end{array}$\bsep{10pt}\\
$\begin{array}{@{}c@{}}(D6)\\[-.5ex] \kappa\neq 1,-2\end{array}$ &
$
\pmatrx{
 \bme_1 & \kappa \bme_2 & \frac{\kappa+1}{2}\bme_3 \\
 \bme_2 & 0 & 0 \\
 \bme_3 & 0 & -\frac{1}{2}\bme_2 }
$ & $
\pmatrx{
 g_{11} & -{\scriptstyle(\kappa +1)} g_{33} & g_{13} \\
 -{\scriptstyle(\kappa +1)} g_{33} & 0 & 0 \\
 g_{13} & 0 & g_{33} }
$ & $
\pmatrx{
 0 & 0 & {\scriptstyle\delta_{\kappa,-1}}f_{13} \\
 0 & 0 & 0 \\
 -{\scriptstyle\delta_{\kappa,-1}}f_{13} & 0 & 0 }
$ & $
\pmatrx{
 h_{11} & 0 & 0\\
 0 & 0 & 0\\
0 & 0 & 0}
$ & $
\pmatrx{
 0 & 0 & {\scriptstyle\delta_{\kappa,-5}}\eta_{13} \\
 0 & 0 & 0 \\
 {\scriptstyle\delta_{\kappa,-5}}\eta_{13} & 0 & 0 }
 $ &
\bsep{10pt}\\
$\begin{array}{@{}c@{}}(D6)\\[-.5ex] \kappa=-2\end{array}$ &
$
\pmatrx{
 \bme_1 & -2 \bme_2 & -\frac{1}{2}\bme_3 \\
 \bme_2 & 0 & 0 \\
 \bme_3 & 0 & -\frac{1}{2}\bme_2 }
$ & $
\pmatrx{
 g_{11} & g_{12} & g_{13} \\
 g_{12} & 0 & 0 \\
 g_{13} & 0 & g_{12}
}
$ & $
\pmatrx{
0 & 0 & 0\\
0 & 0 & 0\\
0 & 0 & 0}
$ & $
\pmatrx{
 h_{11} & 0 & 0\\
 0 & 0 & 0\\
0 & 0 & 0}
$ & $
\pmatrx{
 0 & \eta _{12} & 0 \\
 \eta _{12} & 0 & 0 \\
 0 & 0 & \eta _{12} }
 $ &
 $\det \eta \neq 0$ \bsep{10pt}\\
\hline
\end{tabular}
\end{table}
\end{landscape}
}

{\arraycolsep=3pt
\setlength{\tabcolsep}{-2pt}
\begin{landscape}
\begin{table}[h!]\centering
 \caption{The classification of the bilinear forms generating respective differential $2$-cocycles \eqref{cocycles} and \eqref{cocy} for all non-decomposable non-transitive $4$-dimensional Novikov algebras with non-degenerate bilinear forms $g$ and $\eta$. The Novikov algebras are extracted from the work \cite{FLS}, where
 the corresponding Hamiltonian operators of hydrodynamic type are classified. All the following Novikov algebras are non-abelian and non-associative.}\label{tab2}
 \vspace{1mm}

\begin{tabular}{ c c c c c c c }
\hline type & $\mathcal{B}$ & $g$ & $f$ & $h$ & $\eta$ \\\hline\\[-10pt]

$\begin{array}{@{}c@{}}(N3)\otimes (N6)\\[-.5ex] \kappa=-2\end{array}$ & $
\pmatrx{
\bme_1 & -2 \bme_2 & \bme_3 & -2 \bme_4 \\
 \bme_2 & 0 & \bme_4 & 0 \\
 \bme_3 & -2 \bme_4 & 0 & 0 \\
 \bme_4 & 0 & 0 & 0 }
$ & $
\pmatrx{
 g_{11} & g_{12} & g_{13} & g_{14} \\
 g_{12} & 0 & g_{14} & 0\\
 g_{13} & g_{14} & 0 & 0 \\
 g_{14} & 0 & 0 & 0 \\
}
$ & $
\pmatrx{
 0 & 0 & 0 & 0 \\
 0 & 0 & 0 & 0 \\
 0 & 0 & 0 & 0 \\
 0 & 0 & 0 & 0 \\
}
$ & $
\pmatrx{
 h_{11} & 0 & h_{13} & 0 \\
 0 & 0 & 0 & 0 \\
 h_{13} & 0 & 0 & 0 \\
 0 & 0 & 0 & 0 }
$ & $
\pmatrx{
 0 & \eta _{12} & 0 & \eta _{14} \\
 \eta _{12} & 0 & \eta _{14} & 0 \\
 0 & \eta _{14} & 0 & 0 \\
 \eta _{14} & 0 & 0 & 0 }
$ \bsep{10pt}\\

 & $
\pmatrx{
 \bme_1 & -2\bme_2 & 0 & -\bme_4 \\
\bme_2 & 0 & 0 & 0 \\
\bme_3 & 0 & 0 & -\bme_2 \\
\bme_4 & 0 & 0 & 0 }
$ & $
\pmatrx{
g_{11} & g_{12} & g_{13} & g_{14} \\
 g_{12} & 0 & 0 & 0 \\
 g_{13} & 0 & g_{33} & g_{12} \\
 g_{14} & 0 & g_{12} & 0 }
$& $
\pmatrx{
 0 & 0 & f_{13} & 0 \\
 0 & 0 & 0 & 0 \\
 -f_{13} & 0 & 0 & 0 \\
 0 & 0 & 0 & 0 }
$ & $
\pmatrx{
 h_{11} & 0 & 0 & 0 \\
 0 & 0 & 0 & 0 \\
 0 & 0 & 0 & 0 \\
 0 & 0 & 0 & 0 }
$ & $
\pmatrx{
 0 & \eta _{12} & 0 & 0 \\
 \eta _{12} & 0 & 0 & 0 \\
 0 & 0 & 0 & \eta _{12} \\
 0 & 0 & \eta _{12} & 0 }
$ \bsep{10pt}\\

 & $
\pmatrx{
\bme_1 & -2\bme_2 & 0 & -\bme_4 \\
\bme_2 & 0 & 0 & 0 \\
\bme_3-\bme_4 & 0 & 0 & -\bme_2 \\
\bme_4 & 0 & 0 & 0 }
$ & $
\pmatrx{
g_{11} & g_{12} & g_{13} & 0 \\
 g_{12} & 0 & 0 & 0 \\
 g_{13} & 0 & g_{33} & g_{12} \\
0 & 0 & g_{12} & 0 }
$& $
\pmatrx{
 0 & 0 & f_{13} & 0 \\
 0 & 0 & 0 & 0 \\
 -f_{13} & 0 & 0 & 0 \\
 0 & 0 & 0 & 0 }
$ & $
\pmatrx{
 h_{11} & 0 & 0 & 0 \\
 0 & 0 & 0 & 0 \\
 0 & 0 & 0 & 0 \\
 0 & 0 & 0 & 0 }
$ & $
\pmatrx{
 0 & \eta _{12} & 0 & 0 \\
 \eta _{12} & 0 & 0 & 0 \\
 0 & 0 & \eta _{12} & \eta _{12} \\
 0 & 0 & \eta _{12} & 0 }
$ \bsep{10pt}\\

 & $
\pmatrx{
 \bme_1+\bme_3 & -2 \bme_2 & 0 & -2 \bme_2-\bme_4 \\
 \bme_2 & 0 & 0 & 0 \\
 \bme_3 & 0 & 0 & -\bme_2 \\
 \bme_2+\bme_4 & 0 & 0 & 0 }
$ & $
\pmatrx{
 g_{11} & g_{12} & g_{13} & g_{14} \\
 g_{12} & 0 & 0 & 0 \\
 g_{13} & 0 & 0 & g_{12} \\
 g_{14} & 0 & g_{12} & 0 }
$& $
\pmatrx{
 0 & 0 & 0 & 0 \\
 0 & 0 & 0 & 0 \\
0 & 0 & 0 & 0 \\
 0 & 0 & 0 & 0 }
$ & $
\pmatrx{
 h_{11} & 0 & 0 & 0 \\
 0 & 0 & 0 & 0 \\
 0 & 0 & 0 & 0 \\
 0 & 0 & 0 & 0 }
$ & $
\pmatrx{
 0 & \eta _{12} & 0 & 0 \\
 \eta _{12} & 0 & 0 & 0 \\
 0 & 0 & 0 & \eta _{12} \\
 0 & 0 & \eta _{12} & 0 }
$ \bsep{10pt}\\

$\begin{array}{@{}c@{}}\kappa\geqslant -\frac{1}{2}\\ \kappa\neq 0,1\end{array}$ & $
\pmatrx{
\bme_1 & \kappa \bme_2 & -(\kappa +1) \bme_3 & -2 \bme_4 \\
 \bme_2 & 0 & -(\kappa +1) \bme_4 & 0 \\
 \bme_3 & \kappa \bme_4 & 0 & 0 \\
 \bme_4 & 0 & 0 & 0 }
$ & $
\pmatrx{
 g_{11} & g_{12} & g_{13} & g_{14} \\
 g_{12} & 0 & g_{14} & 0 \\
 g_{13} & g_{14} & 0 & 0 \\
 g_{14} & 0 & 0 & 0 }
$& $
\pmatrx{
 0 & 0 & 0 & 0 \\
 0 & 0 & 0 & 0 \\
0 & 0 & 0 & 0 \\
 0 & 0 & 0 & 0 }
$ & $
\pmatrx{
 h_{11} & 0 & 0 & 0 \\
 0 & 0 & 0 & 0 \\
 0 & 0 & 0 & 0 \\
 0 & 0 & 0 & 0 }
$ & $
\pmatrx{
 0 & 0 & 0 & \eta _{14} \\
 0 & 0 & \eta _{14} & 0 \\
 0 & \eta _{14} & 0 & 0 \\
 \eta _{14} & 0 & 0 & 0 }
$ \bsep{10pt}\\

 & $
\pmatrx{
 \bme_1-2 \bme_2 & \bme_3 & 4\bme_4- \bme_3 & -2 \bme_4 \\
 \bme_2-2 \bme_3 & \bme_4 & -\bme_4 & 0 \\
 \bme_3-2 \bme_4 & 0 & 0 & 0 \\
 \bme_4 & 0 & 0 & 0}
$ & $
\pmatrx{
 g_{11} & g_{12} & g_{13} & g_{14} \\
 g_{12} & g_{13} & g_{14} & 0 \\
 g_{13} & g_{14} & 0 & 0 \\
 g_{14} & 0 & 0 & 0 }
$& $
\pmatrx{
 0 & 0 & 0 & 0 \\
 0 & 0 & 0 & 0 \\
0 & 0 & 0 & 0 \\
 0 & 0 & 0 & 0 }
$ & $
\pmatrx{
 h_{11} & 0 & 0 & 0 \\
 0 & 0 & 0 & 0 \\
 0 & 0 & 0 & 0 \\
 0 & 0 & 0 & 0 }
$ & $
\pmatrx{
 0 & 0 & 0 & \eta _{14} \\
 0 & 0 & \eta _{14} & 0 \\
 0 & \eta _{14} & 0 & 0 \\
 \eta _{14} & 0 & 0 & 0 }
$ \bsep{10pt}\\

\hline
\end{tabular}
\end{table}
\end{landscape}}

\subsection*{Acknowledgements}

I would like to thank Maciej B{\l}aszak, Artur Sergyeyev and Ian Strachan for various conversations concerning theory of higher-dimensional systems and/or topics related to Novikov algebras. I~would also like to thank the anonymous referees for providing to my attention the reference~\cite{FLS} and some important comments.

\pdfbookmark[1]{References}{ref}
\LastPageEnding

\end{document}